\documentclass[11pt]{article}

\usepackage{fullpage}
\usepackage{xspace}
\usepackage{amssymb,amsmath}
\usepackage{epsfig}

\newtheorem{Thm}{Theorem}
\newtheorem{Lem}{Lemma}

\newtheorem{Claim}{Claim}

\newtheorem{Def}{Definition}
\newtheorem{Fact}{Fact}
\newenvironment{proof}{\noindent {\textbf{Proof }}}{\begin{flushright}$\Box$\end{flushright} 
\medskip}

\newcommand{\defeq}{\stackrel{\mathsf{def}}{=}}

\newcommand\mbR{\mbox{$\mathbb{R}$}}

\newcommand\mcA{\mathcal{A}}
\newcommand\mcO{\mathcal{O}}

\newcommand{\cc}{{\mathbb{C}^2}}
\newcommand{\h}{{\mathcal{H}}}
\newcommand{\bit}{{\mathcal{B}}}
\newcommand{\hh}{{\mathcal{H}^{\otimes 2}}}
\newcommand{\hd}{{\mathcal{H}^{\otimes d}}}
\newcommand{\hatt}{{\mathcal{H}^{\otimes t}}}

\newcommand{\td}{{{T}^{\otimes d}}}
\newcommand{\watt}{{{W}^{\otimes t}}}
\newcommand{\watd}{{{W}^{\otimes d}}}

\newcommand{\wattx}{{{W}_x}^{\otimes t}}

\mathchardef\mhyphen="2D

\newcommand{\suppress}[1]{}
\newcommand\COMMENT[1]{}
\newcommand{\prob}{\mathop{\mathsf{Pr}}}

\newcommand\ket[1]{| #1 \rangle}
\newcommand\bra[1]{\langle #1 |}
\newcommand\ketbra[1]{| #1 \rangle \langle #1 |}

\newcommand\norm[1]{\| #1 \|}
\newcommand\abs[1]{| #1 |}

\newcommand\p{\mathsf{P}}
\newcommand\fp{\mathsf{FewP}}
\newcommand\rp{\mathsf{RP}}
\newcommand\up{\mathsf{UP}}
\newcommand\pup{\mathsf{PromiseUP}}
\newcommand\bpp{\mathsf{BPP}}
\newcommand\pp{\mathsf{PP}}
\newcommand\bqp{\mathsf{BQP}}
\newcommand\bqnp{\mathsf{BQNP}}
\newcommand\np{\mathsf{NP}}
\newcommand\uma{\mathsf{UMA}}
\newcommand\ma{\mathsf{MA}}
\newcommand\qma{\mathsf{QMA}}
\newcommand\qcma{\mathsf{QCMA}}
\newcommand\uqma{\mathsf{UQMA}}
\newcommand\uqcma{\mathsf{UQCMA}}
\newcommand\fqma{\mathsf{FewQMA}}
\newcommand\afqma{\mathsf{Alternative\mhyphen FewQMA}}
\newcommand\vfqma{\mathsf{Vector\mhyphen FewQMA}}
\newcommand\alttest{\mathsf{Alternating \, Test}}
\newcommand\wittest{\mathsf{Witness \, Test}}
\newcommand\swaptest{\mathsf{Swap \, Test}}

\newcommand\uqmacp{\mathsf{UQMA \mhyphen CPP}}
\newcommand\sat{\mathsf{SAT}}
\newcommand\locham{\mathsf{Local \, Hamiltonian}}
\newcommand\gnm{\mathsf{Group \, non\mhyphen Membership}}

\newcommand\alt{\mathsf{Alt}}
\newcommand\sym{\mathsf{Sym}}
\newcommand\spn{\mathsf{span}}
\newcommand\sgn{\mbox{sgn}}

\newcommand\accept{\mathsf{Accept}}
\newcommand\reject{\mathsf{Reject}}

\begin{document}
\title{On the power of a unique quantum witness}

\COMMENT{
\author{%
  Rahul Jain$^{1}$ 
  \and 
  Iordanis Kerenidis$^{2}$ 
  \and 
  Greg Kuperberg$^{3}$
  \and
  Miklos Santha$^{2}$
  \and
  Or Sattath$^{4}$
  \and
  Shengyu Zhang$^{5}$
 }
\address{%
  $^{1}$Centre for Quantum Technologies and Department of Computer Science, National University of Singapore
  \and
  $^{2}$CNRS - LRI, Universit\'e Paris-Sud, Orsay, France and Centre for Quantum Technologies, National University of Singapore
  \and
  $^{3}$Mathematics Department, University of California, Davis, USA
  \and
  $^{4}$School of Computer Science and Engineering, The Hebrew University, Jerusalem, Israel
  \and
  $^{5}$Department of Computer Science and Engineering, The Chinese University of Hong Kong
}
\email{%
	rahul@comp.nus.edu.sg
  \and
	jkeren@liafa.jussieu.fr
  \and
  greg@math.ucdavis.edu
  \and
  santha@liafa.jussieu.fr 
  \and
	sattath@cs.huji.ac.il
  \and
  syzhang@cse.cuhk.edu.hk
}
}
\author{ 
Rahul Jain\thanks{Centre for Quantum Technologies and Department of Computer Science, National University of Singapore. Email: {\tt rahul@comp.nus.edu.sg}} \quad 
Iordanis Kerenidis\thanks{CNRS - LIAFA,
Universit\'e Paris Diderot, Paris, France. Email: {\tt jkeren@liafa.jussieu.fr}} \quad 
Greg Kuperberg\thanks{Mathematics Department, University of California, Davis, USA. Email: {\tt   greg@math.ucdavis.edu}} \quad
Miklos Santha\thanks{CNRS - LIAFA,
Universit\'e Paris Diderot, Paris, France and Centre for Quantum Technologies, National University of Singapore. Email: {\tt santha@liafa.jussieu.fr}} \quad \\
Or Sattath\thanks{School of Computer Science and Engineering, The Hebrew University, Jerusalem, Israel. Email: {\tt sattath@cs.huji.ac.il}} \quad
Shengyu Zhang\thanks{Department of Computer Science and Engineering, The Chinese University of Hong Kong. Email: {\tt syzhang@cse.cuhk.edu.hk}}}

\maketitle

\begin{abstract}
In a celebrated paper, Valiant and Vazirani~\cite{VV85} raised the question of whether the difficulty of
$\np$-complete problems was due to the wide variation of the number of witnesses of their instances. They gave a strong negative answer by showing that distinguishing between instances having zero or one witnesses is as hard as recognizing $\np$, under randomized reductions.

We consider the same question in the quantum setting and investigate the possibility of reducing quantum witnesses in the context of the complexity class $\qma$, the quantum analogue of $\np$. The natural way to quantify the number of quantum witnesses is the dimension of the witness subspace  $W$ in some appropriate Hilbert space $\h$.
We present an efficient deterministic procedure
that reduces any problem where the dimension $d$ of $W$ is bounded by a polynomial to a problem with a unique quantum witness. The main idea of our reduction is to consider the Alternating subspace of the tensor power $\hd$. Indeed, the intersection of this subspace with $\watd$ is one-dimensional, and therefore can play the role of the unique quantum witness.
\end{abstract}


\section{Introduction}
One of the most fundamental ideas of modern complexity theory is that, the study of decision making
procedures involving a single party  should be extended to the study of more complex  procedures 
where several parties interact.
The notions of {\em verification} and {\em witness} are at the heart of those complexity classes
whose definition inherently involves interaction.
The complexity classes $\p$ is the set of languages decidable by a polynomial-time deterministic algorithm. Similarly, 
$\bpp$ is the set of promise problems decidable
by a polynomial-time  bounded-error randomized algorithm. 
We can think of such an algorithm
as a verifier acting alone. The simplest interactive extensions of $\p$ and $\bpp$ are
their non-deterministic analogues, respectively $\np$ and $\ma$~\cite{GMR89, Bab85}.
These classes involve also an all powerful
prover that sends a single message which is used by the verifier's  decision making procedure
together with the input. We require that on positive instances there is some message (called 
in that case a witness)
that makes the verifier accept, whereas on negative instances the verifier rejects independently of the message
sent by the prover. In the case of $\ma$ we can fix the permitted error of the verifier, 
rather arbitrarily to any constant, say $1/3$.

Quantum complexity classes are often defined by analogy to their classical counterparts. 
Since quantum computation is inherently probabilistic, the quantum analog 
of $\ma$ is considered to be the right definition of non-deterministic quantum polynomial-time. 
The quantum extension is twofold: 
the verifier has the power to decide promise problems in $\bqp$, quantum polynomial-time, 
and the messages he receives from the
prover are also quantum. Thus, $\qma$ is the set of promise problems such that on positive instances there exists a quantum witness accepted with probability at least $2/3$ by the polynomial-time quantum verifier and on negative instances the verifier accepts every quantum state with probability at most $1/3$. 
While the idea that a quantum state might play the role of a witness goes back to Knill~\cite{Kni96},
the class was formally defined by Kitaev~\cite{KSV02} under the name of $\bqnp$. 
The currently used name $\qma$ was given to the class by Watrous~\cite{Wat00}.
Kitaev has established several
error probability reduction properties of $\qma$, and proved that the $\locham$, the quantum analog
of $\sat$ was complete for it. 
Watrous has shown that $\gnm$ was a problem in $\qma$ and based on this result he has constructed an oracle under which $\ma$ is strictly included in $\qma$. 
Since then, various problems have been proven to be complete for $\qma$~\cite{JWT03, KKR06, Liu06, Kay07, LCV07}. 
A potentially weaker quantum extension of $\ma$, namely $\qcma$, was defined by Aharonov and Naveh~\cite{AN02}: in the case of $\qcma$, the verifier is still a quantum polynomial-time algorithm, but the message of the prover can only be classical.

The number of witnesses for positive instances of problems in $\np$ can be exponentially high.
Also, known $\np$-complete problems have different instances with widely varying numbers of
solutions. In a celebrated paper, Valiant and Vazirani~\cite{VV85} have raised the question of whether the difficulty of
the class $\np$ was due to this wide variation. They gave a strong negative answer to this question in the 
following sense. Let $\up$ be the set of problems in $\np$ where in addition on positive instances there exists a unique witness. We denote by $\pup$ the extention of $\up$ from languages to promise problems.  
The theorem of Valiant and Vazirani states that any problem in $\np$
can be reduced in randomized polynomial-time to a promise problem in $\pup$, or in set theoretical terms,
$\np \subseteq \rp^\pup$, where $\rp$ is the subclass of problems in $\bpp$ where the computation does not
err on negative instances. The complexity class $\up$ has also its importance because of its connection to
one-way functions: worst case one-way functions exist if and only if $\up \neq \p$~\cite{Ko85, GS88}.

In a recent paper Aharonov, Ben-Or, Brand\~{a}o and Sattah~\cite{ABBS08} have asked a similar question for $\ma,~\qcma$ and $\qma$. 
The restriction of the classical-witness classes $\ma$ and $\qcma$ to their unique variants $\uma$ and $\uqcma$ is rather
natural: no change for negative instances, but on positive instances there has to be exactly one witness that makes the verifier accept with probability
at least $2/3$, while all other messages make him accept with probability at most $1/3$. 
The definition of $\uqma$, the unique variant of $\qma$ is the following: there is no change for negative 
instances with respect to $\qma$, but on positive instances there has to be a quantum witness state $\ket{\psi}$ which is accepted 
by the verifier with probability at least $2/3$, whereas all states orthogonal to $\ket{\psi}$ are accepted
with probability at most $1/3$. Aharonov et al. extended the Valiant-Vazirani proof for the classical witness
classes by showing that $\ma \subseteq \rp^\uma$ and $\qcma \subseteq \rp^\uqcma$.
On the other hand, they left the existence of a similar result for $\qma$ as an open problem.

Why is it so difficult to reduce the witnesses to a single witness in the quantum case? The basic idea of Valiant and Vazirani is to use pairwise independent universal hash functions, having polynomial size descriptions, that eliminate independently each witness with some constant probability. The 
size of the original witness set can be guessed approximately by a polynomial-time probabilistic procedure, 
and in case of a correct guess the hashing keeps alive exactly one witness with again some constant probability.
The same idea basically works for $\ma$ and $\qcma$ as long as
one additional difficulty is overcome:
on positive instances there can be exponentially more ``pseudo-witnesses", accepted with probability between $1/3$ and $2/3$, 
than witnesses which are accepted with probability at least $2/3$. In this case, the Valiant--Vazirani proof technique will eliminate with high probability all witnesses before the elimination of the pseudo-witnesses.
The solution of Aharonov et al. for  this problem is to divide the interval $(1/3, 2/3)$ into polynomially many
smaller intervals and to show that there exists at least one interval such that there are approximately as many witnesses accepted with probability within this interval as above it. 

In the quantum case, the set of quantum witnesses can be infinite. 
For a promise problem in QMA, we can suppose without loss of generality that on positive instances there exists a subspace $W$ such that all unit vectors in $W$ are accepted.
The dimension of $W$ could be large and we wish to reduce it to one. 
Aharonov et. al~\cite{ABBS08} considered
the special case where the dimension of $W$ is two. Although classically two witnesses are trivially reducible to the unique witness case, they have shown that the natural generalization of the Valiant--Vazirani construction cannot solve even the two-dimensional quantum witness case.  

Indeed, the natural generalization of the Valiant--Vazirani construction to this situation is to use 
random projections and hope that some one-dimensional subspace of $W$ will be accepted with substantially higher probability than its orthogonal.
A first difficulty is to implement such projections efficiently. 
But more importantly, a random projection would not create a polynomial gap in the acceptance probabilities for the 
pure states of $W$: in fact all states in $W$ which were accepted with exponentially close probabilities, will still be accepted after the random projection  with exponentially close probabilities.  

Here we describe a fundamentally different proof technique to tackle this problem, which is sufficiently powerful to solve the case when the dimension of the witness subspace $W$ is polynomially bounded in the length of the input. This leads us naturally to the quantum analog of the
promise problem class $\fp$. This complexity class was defined by Allender~\cite{All86} as the set of problems
in $\np$ with the additional constraint that there is a polynomial $q$ such that on every positive
instance of length $n$, the number of witnesses is at most $q(n)$. The class $\fp$ was extensively
studied in the context of counting complexity classes~\cite{AR88, Tor90, KSTT92, HJV93, RRW94}.
We define  $\fqma$, the quantum analog of $\fp$, as the set of promise problems in $\qma$ for which there exists  
a polynomial $q$ with the following properties: on negative instances every message
of the prover is accepted by the verifier with probability at most $1/3$; on a positive instance $x$ there exists a subspace 
$W_x$ of dimension between $1$ and $q(|x|)$, such that all pure states in $W_x$ are accepted with probability at least $2/3$, while all pure states orthogonal to $W_x$ are accepted with probability at most $1/3$. Our main theorem 
extends the result of Valiant and Vazirani to this complexity class. 
More precisely, we show that $\fqma$ is deterministic polynomial-time Turing-reducible to $\uqma$. 

\vspace{0.1in}

\noindent {\bf Main Theorem }
$
~~~ \fqma \subseteq \p^\uqma.
$

\vspace{0.1in}

\noindent The first idea to establish this result is that instead of manipulating the states within the original space $\h$ of dimension $K$, we consider its $t$-fold tensor powers $\hatt$. At first glance, this does not seem to be going in the right direction because the dimension of $W^{\otimes t}$ grows as $d^t$, where $d$ is the dimension of the witness space $W$.  
Our second idea is to consider the alternating subspace $\alt$ of $\hatt$ whose dimension is $K \choose t$. The important thing to notice is that the dimension of the intersection $\alt \cap \watt$ is equal to one when $t=d$. The reason is that this intersection is in fact equal to the alternating subspace of $\watt$ whose dimension is $d \choose t$. Therefore, we will choose this one-dimensional subspace as our unique quantum witness.  Of course, we don't know exactly the dimension of $W$, but since we have a polynomial upper bound $q(|x|)$ on it, we just try every
possible value $t$ between 1 and $q(|x|)$.

For a fixed $t$, we would ideally implement $\Pi_\watt  \cdot \Pi_{\alt}$, the product of the projection
to $\alt$ followed by the projection to $\watt$. The reason that this would work is the following. The unique
pure state in $\alt \cap \watt$ (up to a global phase) is clearly accepted with probability 1. On the other hand, 
we claim that any state $\ket{\phi}$ orthogonal to that is rejected with probability 1. Indeed, $\ket{\phi}$ 
can be decomposed as $\ket{\phi_1} +  \ket{\phi_2}$, where $\ket{\phi_1} \in  \alt^\bot$ 
and $\ket{\phi_2} \in \watt^\bot$. Therefore $\ket{\phi_1}$ is rejected by $\Pi_{\alt}$ and
$\ket{\phi_2}$ is rejected by $\Pi_\watt$. This implies the claim since 
we can show that the two projectors actually commute.

We can efficiently implement $ \Pi_{\alt}$ by a procedure we call the $\alttest$. A similar procedure to ours, implementing efficiently the projection to the symmetric subspace $\sym$ of $\hatt$, was proposed by Barenco et al.~\cite{BBDEJM97} as the basis
of a method for the stabilization of quantum computations. In fact, in the two-fold tensor product case, the two procedures coincide and become the well know $\swaptest$ which was used by Buhrman et al.~\cite{BCWW01} for deciding if two given pure states are close or far apart. 

We can't implement $\Pi_\watt$ exactly, but we can approximate it efficiently by a procedure called the
$\wittest$. This test just applies 
independently to all the $t$ components of the state the procedure at our disposal which decides 
in $\h$ whether a state is a witness or not, and accepts if all applications accept.
There is only one difficulty left: since $ \Pi_{\alt}$
and the $\wittest$ don't necessarily commute, our previous argument which showed that states
in $\watt^\bot$ were rejected with probability 1 doesn't work anymore. We overcome this difficulty
by showing that the commutativity of the two projections implies that the projections to $\alt$ of
such states are also in $\watt^\bot$, and therefore get rejected with high probability by the $\wittest$.

An interesting feature of our reduction is that it is deterministic,
while the Valiant-Vazirani procedure is probabilistic. It is fair to
say though that classically the witnesses can all be enumerated when
their number is bounded by a polynomial. Therefore, in that case, the
reduction can also be done deterministically, implying that $\fp
\subseteq \p^{\pup}$. We believe that reducing $\qma$ to a unique
witness, which this paper leaves as an open question, will require a
probabilistic or a quantum procedure.

The rest of the paper is structured as follows. In Section~\ref{prelim} we state some facts about the interaction
of the tensor products of subspaces with the alternating subspace. In Section~\ref{classes} 
we define the complexity classes we are concerned with. We give two definitions for
$\fqma$ and show that they are equivalent. Section~\ref{Theorem} is entirely 
devoted to the proof of our main result. Finally in the Appendix~\ref{sec:third} we consider a third definition and show a weak equivalence with the previous ones. 

The results in the paper appeared initially as Arxiv preprint quant-ph/0906.4425 by Jain, Kerenidis, Santha and Zhang. Some of this work was done independently by Sattath and Kuperberg and an initial result in  that direction appeared in p.30 of \cite{Sattath}. 

\section{Preliminaries} \label{prelim}

In this section we present definitions and lemmas that we will need in the proof of our main result. 

We represent by $[t]$  the set $\{1,2, \ldots, t\}$. For a Hilbert space $\h$, we denote by $\dim(\h)$ the dimension of $\h$. 
For a subspace $S$ of $\h$, let $S^\bot$ represent the subspace of $\h$ orthogonal to $S$,
and let $\Pi_S$ denote the projector onto $S$. For subspaces $S_1,S_2$ of $\h$, 
their {\em direct sum} $S_1 + S_2$ is defined as $\spn(S_1 \cup S_2)$, and when $S_1,S_2$ are orthogonal subspaces, we denote their (orthogonal) direct sum  by  $S_1 \oplus S_2$. The following relations are standard.
\begin{Fact} \label{fact:stand}
\begin{enumerate}
\item Let $S_1, S_2$ be subspaces of a Hilbert space $\h$. Then $(S_1 \cap S_2)^\bot = S_1^\bot + S_2^\bot $.
\item 
Let $S_1, S_2$ be subspaces of Hilbert spaces $\h_1$, $\h_2$ respectively.
Then, $(S_1 \otimes S_2)^\bot = \\(S_1^\bot \otimes \h_2) + (\h_1 \otimes S_2^\bot)
 ~=~(S_1^\bot \otimes \h_2) \oplus (S_1 \otimes S_2^\bot)$. 
\end{enumerate}
\end{Fact}

Let $\bit$ represent the two-dimensional complex Hilbert space and let $\{\ket{0}, \ket{1}\}$ be the computational basis for $\bit$. For a natural number $k$, the computational basis of $\bit^{\otimes k}$ (the $k$-fold tensor of $\bit$) consists of $\{\ket{r}: r \in \{0,1\}^k\}$,
where $\ket{r}$ denotes the tensor product $\ket{r_1} \otimes \ldots \otimes \ket{r_k}$ for the $k$-bit string $r = r_1 \ldots r_k$.
Fix $k$ and let $\h$ denote $\bit^{\otimes k}$ and let $K = 2^k$. By a {\em pure state} in $\h$, we mean a unit vector in $\h$. A {\em mixed state} or just {\em state} is a positive semi-definite operator in $\h$ with trace $1$. We refer the reader to the text~\cite{NielsenC00} for concepts related to quantum information theory. For a natural number $t \in [K]$,  we will think of states of $\hatt$ as consisting of $t$  registers, 
where the content of each register is a state with support in $\h$.

We will consider the interaction of $\watt$, where $W$ is a $d$-dimensional subspace of $\h$ for some $d$ satisfying $2 \leq t \leq d \leq K$, with the alternating and symmetric subspaces of $\hatt$.
Let $S_t$ denote the set of all permutations $\pi : [t] \rightarrow [t]$. For a permutation $\pi \in S_t$, let the unitary operator $U_{\pi}$, acting on $\hatt$, be given by 
$
\ket{s_{\pi(1)}} \otimes  \ldots \otimes \ket{s_{\pi(t)}} \enspace .$

For permutations $\pi_1, \pi_2$, let $\pi_1 \circ \pi_2$ represent their composition. It is easily seen 
that $U_{\pi_1 \circ \pi_2} = U_{\pi_1} U_{\pi_2}$. For distinct $i,j \in [t]$, let $\pi_{ij}$ be the transposition
of $i$ and $j$. 
For all distinct $i, j\in [t]$, the symmetric subspace of $\watt$ 
with respect to $i$ and $j$ is given by  \newline $\sym_{ij}^{\watt} = \{\ket{\phi}\in W^{\otimes t}: U_{\pi_{ij}}\ket{\phi} = \ket{\phi}\}$, 
and the {\em symmetric subspace} of $\watt$ is defined as $\sym^{\watt} = \cap_{i \neq j} \sym_{ij}^{\watt}$.
Similarly, for all distinct $i, j\in [t]$, the alternating subspace of $\watt$ with respect to $i$ and $j$ is
defined as $ \alt_{ij}^{\watt} =   \{\ket{\phi}\in W^{\otimes t}: U_{\pi_{ij}}\ket{\phi} = -\ket{\phi}\}$,
and the {\em alternating subspace} of $\watt$ is defined as
$\alt^{\watt} = \cap_{i \neq j} \alt_{ij}^{\watt}$.

The subspaces $\sym^{W^{\otimes t}}$ and 
$ \alt^{W^{\otimes t}}$ are of 
dimension $d+t-1 \choose t$ and $d \choose t$ respectively~\cite{Bhatia}. 
In particular, $\alt^{\h^{\otimes 2}}$ and $\sym^{\h^{\otimes 2}}$ have respective dimensions $K+1 \choose 2$ and $K \choose 2$  and since they are orthogonal, we have 
$\hh = \alt^{\h^{\otimes 2}} \oplus \sym^{\h^{\otimes 2}}$.
This implies that for every distinct $i,j \in [t],\;$ we have  $\alt_{ij}^{\watt} \oplus \sym_{ij}^{\watt} = \watt$.
It follows that
\begin{Claim}\label{claim:Tdecompij}
 $(\alt^{\watt})^\bot \cap \watt =  \sum_{i \neq j} \sym^{\watt}_{ij}  \enspace .$
\end{Claim}
\begin{proof}
Since $\alt_{ij}^{\watt} \oplus \sym_{ij}^{\watt} = \watt$,
we have $(\alt_{ij}^{\watt})^\bot = \sym_{ij}^{\watt} \oplus (\watt)^\bot$. Therefore
\begin{eqnarray*}
(\alt^{\watt})^\bot \cap \watt
 & = & ((\cap_{i \neq j} \alt_{ij}^{\watt})^\bot ) \cap \watt \; \mbox{(from def. of $\alt^\watt$)} \\
& = & (\sum_{i \neq j} (\alt_{ij}^{\watt})^\bot ) \cap \watt  \quad \mbox{(from Fact~\ref{fact:stand})}\\ 
& = & ( \sum_{i \neq j} (\sym^{\watt}_{ij} \oplus (\watt)^\bot) )\cap \watt \\
& = & ( (\sum_{i \neq j} \sym^{\watt}_{ij} ) \oplus (\watt)^\bot )\cap \watt \\
& = &  \sum_{i \neq j} \sym^{\watt}_{ij} \enspace .
\end{eqnarray*}
The last equality holds since  $(\sum_{i \neq j} \sym^{\watt}_{ij} ) \subseteq \watt$.
\end{proof}

Note that for $W=\h$ the claim states that $(\alt^{\hatt})^\bot  =  \sum_{i \neq j} \sym^{\hatt}_{ij} $.
For us, a particularly important case is when the number of registers $t$ is equal to $d$, the
dimension of the subspace $W$. Then the alternating subspace  $ \alt^{W^{\otimes d}}$ is 
one-dimensional. 
Let $\{ \ket{\psi_1}, \ldots, \ket{\psi_d} \} $ be any orthonormal basis of $W$, and
let the vector $\ket{W_{alt}} \in \watd$ be defined as
 $\ket{W_{alt}} = \frac{1}{\sqrt{d!}}\sum_{\pi\in S_d} \mbox{sgn}(\pi) \; 
 U_\pi \ket{\psi_1} \ldots \ket{\psi_d}$, where $\mbox{sgn}(\pi)$ denotes the sign of the permutation $\pi$.
The following claim states  that $\ket{W_{alt}}$ spans the one-dimensional
subspace $\alt^{W^{\otimes d}}$. 
This  immediately implies that
$\ket{W_{alt}}$ is independent of the choice of the basis (up to a global phase). 
\begin{Claim} \label{claim:unique}
$ \alt^{W^{\otimes d}} = \spn\{\ket{W_{alt}}\} \enspace .$
\end{Claim}
\begin{proof}
We show that $\ket{W_{alt}} \in \alt^{W^{\otimes d}}$. This implies the statement since 
$\dim( \alt^{W^{\otimes d}}) = {d \choose d} = 1$.
For any distinct $i,j \in [d]$ we show $\ket{W_{alt}} \in \alt^{W^{\otimes d}}_{ij}$. For a permutation $\pi \in S_d$, we set $\pi' = \pi_{ij} \circ \pi $. 
We then have
\begin{eqnarray*} 
U_{\pi_{ij}} \ket{W_{alt}} 
& = & U_{\pi_{ij}} \frac{1}{\sqrt{d!}}\sum_{\pi\in S_d} \mbox{sgn}(\pi) \; U_\pi \ket{\psi_1} \otimes \ldots \otimes \ket{\psi_d} \\
& = & \frac{1}{\sqrt{d!}}\sum_{\pi\in S_d} \mbox{sgn}(\pi) \; U_{\pi_{ij} \circ \pi} \ket{\psi_1} \otimes \ldots \otimes \ket{\psi_d}  \\
& = & \frac{1}{\sqrt{d!}}\sum_{\pi'\in S_d} \mbox{sgn}(\pi_{ij}^{-1} \circ \pi') \; U_{\pi'} \ket{\psi_1} \otimes \ldots \otimes \ket{\psi_d} \\
& = & (-1) \cdot \frac{1}{\sqrt{d!}}\sum_{\pi'\in S_d} \mbox{sgn}( \pi') \; U_{\pi'} \ket{\psi_1} \otimes \ldots \otimes \ket{\psi_d} \\
& = & - \ket{W_{alt}} \enspace ,
\end{eqnarray*} 
where we used that $\mbox{sgn}(\pi_{ij}^{-1} \circ \pi') = -\mbox{sgn}(\pi')$.
\end{proof}

Next, we show that the projections on the spaces $\alt^{\hatt}$ and
$\watt$ commute for any $2 \leq t \leq d$. 
\begin{Claim} \label{claim:commute}
Let $2 \leq t \leq d$. \\
Then, $ \Pi_{\alt^{\hatt}} \cdot \Pi_{\watt} = \Pi_{\watt} \cdot \Pi_{\alt^{\hatt}} \enspace .$
\end{Claim}
\begin{proof}
Set $T =  (\alt^{\watt})^\bot \cap \watt $. Then $\watt = \alt^{\watt}
\oplus T$ and  hence, $\Pi_\watt = \Pi_{\alt^{\watt}} +\Pi_{T}$.
We have
\begin{eqnarray*}
T & = & ( \alt^{W^{\otimes t}} )^{\bot} \cap \watt \\
& = & \sum_{i \neq j} \sym^{\watt}_{ij} \quad \mbox{(from Claim~\ref{claim:Tdecompij}) } \\
& \subseteq &  \sum_{i \neq j} \sym^{\h^{\otimes t}}_{ij} \quad 
\mbox{ (by definition)}\\
& = & (\alt^\hatt)^\bot \quad \mbox{(from Claim~\ref{claim:Tdecompij}) } \enspace .
\end{eqnarray*}
This implies that $  \Pi_{\alt^\hatt} \cdot \Pi_{T} = \Pi_{T} \cdot \Pi_{\alt^\hatt}  = {\bf 0} .$ 
Also, $\alt^{\watt} \subseteq  \alt^\hatt $ since $\alt^{\watt} =
\alt^{\hatt} \cap \watt$. Therefore, we have 
$$\Pi_{\alt^\hatt} \cdot \Pi_\watt  =   \Pi_{\alt^\hatt} \cdot
(\Pi_{\alt^\watt} + \Pi_{T}) = \Pi_{\alt^\watt} \enspace .$$ 
Similarly
$$\Pi_\watt  \cdot \Pi_{\alt^\hatt}  = \Pi_{\alt^\watt} \enspace .$$ 
Hence $\Pi_{\alt^\hatt} \cdot \Pi_\watt = \Pi_\watt  \cdot
\Pi_{\alt^\hatt}$.
\end{proof}

This commutativity relation enables us to derive the following property 
\begin{Claim}\label{claim:projection}
For any state $\ket{\phi} \in (\alt^{\watd})^\bot$, we have  
$\Pi_{\alt^{\hd}} \ket{\phi} \in ( \watd )^\bot$. 
\end{Claim}

\begin{proof}
First note that $( \alt^{\watd} )^\bot = ( \alt^{\hd} \cap \watd )^\bot$
and by Fact 1, $( \alt^{\watd} )^\bot = ( \alt^{\hd}  )^\bot + (  \watd )^\bot$. 
Hence we can decompose $\ket{\phi}$ as $ \ket{\phi_1} + \ket{\phi_2}$, 
where $\ket{\phi_1} \in ( \alt^{\hd}  )^\bot$ 
and $\ket{\phi_2} \in (  \watd )^\bot$. 
As $\Pi_{\alt^{\hd}} \ket{\phi_1} = {\bf 0}$, it suffices to show that 
$\Pi_{\alt^{\hd}} \ket{\phi_2} \in ( \watd )^\bot$. For this, we prove that
$\Pi_{(\watd)^\bot}  \cdot \Pi_{\alt^\hd} \ket{\phi_2} = \Pi_{\alt^\hd} \ket{\phi_2} $.

Claim~\ref{claim:commute} implies that 
$$\Pi_{\alt^\hd} \cdot  \Pi_{(W^{\otimes d})^{\bot}}  =  \Pi_{(\watd)^\bot}  \cdot \Pi_{\alt^\hd}  \enspace .$$ 
Also,  $\ket{\phi_2}  = \Pi_{(\watd)^\bot} \ket{\phi_2}$ since $\ket{\phi_2} \in (  \watd )^\bot$.
Therefore we can conclude by the following equalities:
\begin{eqnarray*}
\Pi_{(\watd)^\bot}  \cdot \Pi_{\alt^\hd} \ket{\phi_2}   =   \Pi_{\alt^\hd}   \cdot  \Pi_{(\watd)^\bot} \ket{\phi_2} 
 =  \Pi_{\alt^\hd} \ket{\phi_2} \enspace .
\end{eqnarray*}
\end{proof}


\section{Complexity classes} \label{classes}
In this section we define the relevant complexity classes and state the facts needed about them.
For a quantum circuit $V$, we let $V$ also represent the unitary transformation corresponding to the circuit.
We call a {\em verification procedure} a  family 
of  quantum circuits $\{V_x : x \in \{0,1\}^*\}$ 
uniformly generated in polynomial-time, together
with polynomials $k$ and $m$ 
such that $V_x$ acts on $k(|x|) +  m(|x|)$ qubits. We refer to the first $k(|x|)$ qubits 
as {\em witness qubits} and to the last $m(|x|)$ qubits as  {\em auxiliary qubits}. To simplify notation, 
when the input $x$ is implicit in the discussion, we 
refer to $k(|x|)$ by $k$, and to $m(|x|)$ by $m$. We will make repeated use of the following projections
in  $\bit^{\otimes (k+m)}$:
$$
\Pi_{acc} = \ketbra{1} \otimes I_{k+m-1}, \ \ \  \Pi_{init} = I_k \otimes \ketbra{0^m},
$$
where $I_n$ is the identity operator on $n$ qubits.
We will also make use of the operator $\Pi_x$ defined as $\Pi_x =  \Pi_{init} V_x^\dag \Pi_{acc}  V_x \Pi_{init}$. It is easy to see that $\Pi_x$ is positive semi-definite.

Given a verification procedure, on input $x$, a Quantum Merlin-Arthur protocol proceeds in the following way: the prover Merlin sends a pure state $\ket{\psi} \in \bit^{\otimes k}$, the {\em witness},
to the verifier Arthur, who then applies 
the circuit $V_x$ to $\ket{\psi} \otimes \ket{0^{m}}$, and accepts  
if the measurement of the first qubit of the result gives $1$. We will denote the probability that Arthur accepts $x$ with witness $\ket{\psi}$ by $\prob[V_x \mbox{ outputs } \accept \mbox{ on } \ket{\psi}]$, which is equal to 
$|| \Pi_{acc} V_x (\ket{\psi} \otimes \ket{0^{m}}) ||^2 .$

A promise problem is a tuple $L=(L_{yes}, L_{no})$ with
$L_{yes} \cup L_{no} \subseteq \{0,1\}^*$ and $L_{yes} \cap L_{no} = \emptyset$. We now define the following complexity classes.
\begin{Def}
A promise problem $L=(L_{yes}, L_{no})$ is in the complexity class 
{\em Quantum Merlin-Arthur} $\qma$ if there exists a verification procedure
$\{V_x : x \in \{0,1\}^*\}$ with polynomials $k$ and $m$ such that 
\begin{enumerate}
\item 
for all $x \in L_{yes}$, there exists  a witness $\ket{\psi} $, 
such that 
$|| \Pi_{acc} V_x (\ket{\psi} \otimes \ket{0^{m}}) ||^2 \geq 2/3,$
\item
for all $x \in L_{no}$, and for all witnesses $\ket{\psi}$, 
$|| \Pi_{acc} V_x (\ket{\psi} \otimes \ket{0^{m}}) ||^2 \leq 1/3.$
\end{enumerate}
\end{Def}

For a promise problem in $ \qma$, we can suppose without loss of generality that on positive instances there exists a subspace $W$ such that all unit vectors in $W$ are accepted. Hence, by putting a polynomial upper bound on the dimension of this subspace, we derive the following definition:

\begin{Def} \label{def:fqma}
Let $c,w,s : \mathbb{N} \rightarrow [0,1]$ be polynomial-time computable functions such that $c(n) > \max\{w(n), s(n)\}$ for all $n \in \mathbb{N}$. A promise problem $L=(L_{yes}, L_{no})$ is in the complexity class 
{\em Few Quantum Merlin-Arthur}
$\fqma(c,w,s)$ if there exists a verification procedure
$\{V_x : x \in \{0,1\}^*\}$ with polynomials $k$ and $ m$, and a polynomial $q$ such that 
\begin{enumerate}
\item 
for all $x \in L_{yes}$ there exists a subspace $W_x$ of $ \bit^{\otimes k}$
with $\dim (W_x) \in [q(|x|)]$ such that 
\begin{enumerate} 
\item for all
witnesses $\ket{\psi} \in W_x,$
$|| \Pi_{acc} V_x (\ket{\psi} \otimes \ket{0^{m}}) ||^2 \geq c(|x|),$
and 
\item for all witnesses
$\ket{\phi} \in {W_x}^{\bot},$
$|| \Pi_{acc} V_x (\ket{\phi} \otimes \ket{0^{m}}) ||^2 \leq w(|x|),$
\end{enumerate}
\item
for all $x \in L_{no}$ and for all pure states
$\ket{\psi} \in \bit^{\otimes k}$,
$|| \Pi_{acc} V_x (\ket{\psi} \otimes \ket{0^{m}}) ||^2 \leq s(|x|).$
\end{enumerate}
\end{Def}

\begin{Def}
The class {\em Few Quantum Merlin-Arthur}
 $\fqma$ is $\fqma(2/3, 1/3, 1/3)$.
\end{Def}

We next provide an alternative definition of $\fqma(c,w,s)$ and show that the two definitions are equivalent.

\begin{Def} \label{def:altfqma}
Let $c,w,s : \mathbb{N} \rightarrow [0,1]$ be polynomial-time computable functions such that $c(n) > \max\{w(n), s(n)\}$ for all $n \in \mathbb{N}$. A promise problem $L=(L_{yes}, L_{no})$ is in the complexity class 
\newline $\afqma(c,w,s)$ if there exists a verification procedure
$\{V_x : x \in \{0,1\}^*\}$ with polynomials $k$ and $m$, and a polynomial $q$ such that 

\begin{enumerate}
\item 
for all $x \in L_{yes}$, the number of eigenvalues of $\Pi_x$ which are at least $c(|x|)$ is in $[q(|x|)]$, 
and no eigenvalue of $\Pi_x$ is in the open interval $(w(|x|), c(|x|))$.
\item
for all $x \in L_{no}$, all eigenvalues of $\Pi_x$ are at most $s(|x|)$.
\end{enumerate}
\end{Def}

We prove the following equivalence between the two definitions.
\begin{Thm}\label{thm:equivalence} 
Let $c,w,s : \mathbb{N} \rightarrow [0,1]$ be polynomial-time computable functions such that $c(n) > \max\{w(n), s(n)\}$ for all $n \in \mathbb{N}$. Then 
$\fqma(c,w,s) = \afqma(c,w,s).$
\end{Thm}

\begin{proof}
\noindent {\bf Part 1 (Definition~\ref{def:altfqma} $\Rightarrow$ Definition~\ref{def:fqma}):} 
Let $L$ be a promise problem $\afqma(c,w,s)$ 
with some verification procedure $\{V_x\}$ and polynomial $q$ according to Definition~\ref{def:altfqma}.
We show that $\{V_x\}$ and $q$ satisfy also Definition~\ref{def:fqma} with the same parameters.

We first consider the case $x \in L$. 
For every $\ket{u} \in \bit^{\otimes (k+m)}$, we have
\begin{equation} 
\label{eq:cent} 
\norm{\Pi_{acc}V_x \Pi_{init} \ket{u}}^2 
= \bra{u} \Pi_x \ket{u} \enspace .
\end{equation}
It is easy to check that all eigenvectors of $\Pi_x$ are also eigenvectors of $ \Pi_{init}$. The $1$-eigenvectors of the projector $\Pi_{init}$ are of the form $\ket{u} \otimes \ket{0^m}$ with $\ket{u} \in \bit^{\otimes k}$, and any vector orthogonal to these is an eigenvector with eigenvalue 0. 
Let $r$ be the number of eigenvalues of $\Pi_x$ that are at least $c(|x|)$, by hypothesis $r \in [q(|x|)]$.
Let $\{\ket{v_i} \otimes \ket{0^m} : i \in [r]\}$ be a set of orthonormal eigenvectors of $\Pi_x$ with 
respective eigenvalues $\{\lambda_i \geq c(|x|): i \in [r]\}$, we set 
$W_x = \spn ( \{ \ket{v_i} : i \in [r] \})$.
Then all remaining $2^{k+m} - r$ eigenvalues of $\Pi_x$ are less than or equal to $w(|x|)$.
Let $\{\ket{v_i} \otimes \ket{0^m} : i \in \{r+1, \ldots, 2^{k+m} \}\}$ be a set of 
orthonormal eigenvectors with such eigenvalues. It is clear then that 
$W_x^{\bot} = \spn ( \{ \ket{v_i} : r < i \leq 2^{k+m} \} )$.

We consider a pure state 
$\ket{\psi} = \sum_{i \in [r]} \alpha_i \ket{v_i}$ in $W_x$. Then,
\begin{eqnarray*}
\norm{\Pi_{acc}V_x(\ket{\psi} \otimes \ket{0^m})}^2
& =&
 \norm{\Pi_{acc}V_x \Pi_{init}(\ket{\psi} \otimes \ket{0^m})}^2 \\
& = & ( \bra{\psi} \otimes \bra{0^m} ) (\Pi_x  \ket{\psi} \otimes \ket{0^m} ) \\
& =& (\bra{\psi} \otimes \bra{0^m}) ((\sum_{i \in [r]} 
\lambda_i  \alpha_i \cdot \ket{v_i}) \otimes \ket{0^m}) \\
& = & \sum_{i \in [r]} \abs{\alpha_i}^2 \cdot \lambda_i \; \geq \; c(|x|) \enspace .
\end{eqnarray*}
If $\ket{\phi} \in W_x^\bot$ is a pure state then by similar arguments  we get 
$\norm{\Pi_{acc}V_x(\ket{\phi} \otimes \ket{0^m})}^2 \leq w(|x|).$

When $x \notin L$, condition 2 of Definition~\ref{def:fqma} gets satisfied analogously 
from condition 2 of Definition~\ref{def:altfqma}.

\vspace{0.2in}

\noindent {\bf Part 2 (Definition~\ref{def:fqma} $\Rightarrow$ Definition~\ref{def:altfqma}):} Let 
$L \in \fqma(c,w,s)$ with some verification procedure $\{V_x\}$ and polynomial $q$ 
according to Definition~\ref{def:fqma}. We claim that $\{V_x\}$ and $q$ satisfy also Definition~\ref{def:altfqma}
with the same parameters.

First consider the case $x \in L$. The cardinality of the dimension of the subspace of witnesses 
$W_x$ in $\bit^{\otimes k}$ is in $[q(|x|)$
by hypothesis. We set
\begin{eqnarray*}
W_c & = & \spn\{\ket{v} \in \bit^{\otimes k} : \ket{v} \otimes \ket{0^m} 
\mbox{ is  an eigenvector }  
 \mbox{of $\Pi_x$ with eigenvalue } \geq \; c(|x|) \} \enspace \\ 
\mbox{and} && \\
W_w & = & \spn\{\ket{v} \in \bit^{\otimes k} : \ket{v} \otimes \ket{0^m} \mbox{ is  an eigenvector }  
 \mbox {of $\Pi_x$ with eigenvalue } > \; w(|x|) \} \enspace
\end{eqnarray*}
We will show that $\dim(W_x) = \dim(W_c)$ and that $W_c = W_w$, from  which the claim
follows. 
For this, it is sufficient to prove that
$\dim(W_c) = \dim(W_x) =  \dim(W_w)$, since clearly $W_c$ is a subspace of $W_w$.

First observe that the definitions of $W_c$ and $W_w$ imply that 
\begin{eqnarray*}
W_c^\bot & = & \spn\{\ket{v} \in \bit^{\otimes k} : \ket{v} \otimes \ket{0^m} \mbox{ is  an eigenvector } 
\mbox {of $\Pi_x$ with eigenvalue } < \; c(|x|) \} \enspace \\
\mbox{and} & & \\
W_w^\bot & = & \spn\{\ket{v} \in \bit^{\otimes k} : \ket{v} \otimes \ket{0^m} \mbox{ is  an eigenvector } 
 \mbox {of $\Pi_x$ with eigenvalue } \leq \; w(|x|) \} \enspace .
\end{eqnarray*}
Let us suppose that  $\dim(W_x) < \dim(W_c)$. Then there exists a vector $\ket{u}$  in 
$W_c \cap W_x^\bot$. Since $\ket{u} \in W_c$, using arguments as in Part 1 above, we have 
$\norm{\Pi_{acc}V_x(\ket{u} \otimes \ket{0^m})}^2 \geq c(|x|)$. 
However, since $\ket{u} \in W_x^\bot$,  from condition 1(b) of Definition~\ref{def:fqma} we have 
$\norm{\Pi_{acc}V_x(\ket{u} \otimes \ket{0^m})}^2 \leq w(|x|) < c(|x|)$ which is a contradiction. 
We similarly reach a contradiction assuming $\dim(W_x) > \dim(W_c)$ and hence $\dim(W_x) = \dim(W_c)$.

The equality $\dim(W_w) = \dim(W_x)$ can be proven by an argument analogous to the 
proof of $\dim(W_x) =  \dim(W_c)$.
 
In the case $x \notin L$,  assume for contradiction that there is an eigenvalue 
$\lambda > s(|x|)$ of $\Pi_x$ with eigenvector $\ket{v} \otimes \ket{0^m}$. Then as before,
\begin{eqnarray*}
\norm{\Pi_{acc}V_x(\ket{\psi} \otimes \ket{0^m})}^2
 & =&
\norm{\Pi_{acc}V_x \Pi_{init} (\ket{v} \otimes \ket{0^m})}^2 \\
& = &  (\bra{v} \otimes \bra{0^m}) \Pi_x (\ket{v} \otimes \ket{0^m}) \; = \; \lambda \; > \; s(|x|) \enspace ,
\end{eqnarray*}
which contradicts condition 2 of Definition~\ref{def:fqma}.
\end{proof}

The alternative definition of $\fqma(c,w,s)$ is useful in arriving at the following strong error probability reduction theorem whose proof follows very similar lines as the QMA strong error probability reduction proof in Marriott and Watrous~\cite{MW05} and hence is skipped.

\begin{Thm}[Error reduction] \label{thm:amplify} 
Let $c,w,s : \mathbb{N} \rightarrow [0,1]$ be polynomial-time computable functions such that for some polynomial $p$, for all $n$, they satisfy
$c(n) > \max\{w(n), s(n)\} + 1/p(n)$.  Let $r$ be any polynomial. Then
for any $L \in \fqma(c,w,s)$ having a verification procedure with
polynomials $k,m,q$, there exists a verification procedure for $L$
with parameters $(c',w',s') = (1-2^{-r},2^{-r}, 2^{-r})$ and
polynomials $k'=k$,  $m' = poly(m, r)$ and $q'=q$.
\end{Thm}


\begin{Def}
A promise problem $L=(L_{yes}, L_{no})$ is in the complexity class 
{\em Unique Quantum Merlin-Arthur} $\uqma$ if $L$ is in $\fqma$ with the additional constraint that
for all $x \in L_{yes}$, the subspace $W_x$ of witnesses in the definition of $\fqma$ is one-dimensional.
\end{Def}

\begin{Def} 
\label{def:uqmacp}
$\uqmacp$ is the promise problem  $ (\uqmacp_{yes}, \uqmacp_{no})$
where the elements of $\uqmacp_{yes} \cup \uqmacp_{no} $ are descriptions of quantum circuits $V$ with $k$ witness qubits and $m$ auxiliary qubits, such that
\begin{enumerate}
\item $V \in \uqmacp_{yes}$ if
\begin{enumerate} 
\item there exists a witness
$\ket{\psi}$ such that 
$\prob [V \mbox{ outputs } \accept \mbox{ on } \ket{\psi} ] \geq 2/3,$ 
\item for all witnesses
$\ket{\phi}$ 
orthogonal to $\ket{\psi}$,
$\prob [V \mbox{ outputs } \accept \mbox{ on } \ket{\phi} ] \leq 1/3,$
\end{enumerate}
\item
$V \in \uqmacp_{no}$ if for all pure states 
$\ket{\psi}$, $\prob [V \mbox{ outputs } \accept \mbox{ on } \ket{\psi}] \leq 1/3.$
\end{enumerate}
\end{Def}
It is easily verified that $\uqmacp$ is the canonical $\uqma$-complete promise problem.


\section{Reducing the dimension of the witness subspace}\label{Theorem}

This section will be entirely devoted to the proof of our main result.

\begin{Thm} {\bf (Main Theorem)} \label{main}
$\fqma \subseteq \p^{\uqma}.$
\end{Thm}

\begin{proof} 

Let $L \in \fqma$ have a verification procedure $\{V'_x : x \in \{0,1\}^*\}$ with polynomials $k,m',q$. Let $r$ be a polynomial such that $q 2^{-r} \leq 1/3$. Then, we know from Theorem \ref{thm:amplify}
that $L \in \fqma(1 - 2^{-r}, 2^{-r}, 2^{-r})$ with verification procedure $\{ V_x : x \in \{0,1\}^*$ and polynomials $k,m,q$.

Our goal is to describe a deterministic polynomial-time algorithm $\mcA$, with access to the oracle $\mcO$ for the promise problem $\uqmacp$, that decides the promise problem $L$. 
In high level, our algorithm works in the following way. On input $x$ and for all $t \in [q(|x|)]$, $\mcA$ calls $\mcO$ with a quantum circuit $A_x^t$ that uses $t \cdot k$ witness qubits and $t \cdot m$ auxiliary qubits, and outputs $\accept$ if and only if the witness has the following two properties: first, it belongs to the alternating subspace of $\h^{\otimes t}$ and second the circuit $V_x$, when performed on each of the $t$ registers separately, outputs $\accept$ on all of them. $\mcA$ accepts iff for any $t$, oracle $\mcO$ accepts. We will prove that for $x\in L_{yes}$, we have $A_x^d \in \uqmacp_{yes}$, where $d=\dim(W_x)$. Hence $\mcO$ accepts $A_x^t$ and therefore $\mcA$ accepts. On the other hand, for $x\in L_{no}$, we show that for all $t \in [q(|x|)], A_x^t \in \uqmacp_{no}$ and hence $\mcA$ rejects.   

We first describe in detail the $\alttest$ and the $\wittest$ that appear in the algorithm. In our descriptions below $k,m,q,r$ represent the integers $k(|x|), m(|x|),q(|x|)$ and $r(|x|)$ respectively.


\subsection{Alternating Test}

Let $\h$ be the Hilbert space $\bit^{\otimes k}$ and let $t \in [2^k]$. Let us fix some polynomial-time computable bijection between the set $[t!]$ and the set of permutations $S_t$. Let $\mathcal{P}_t$ be the $(t!)-$dimensional Hilbert space spanned by vectors $\ket{i}$, for $i \in [t!]$. We will use the elements of $S_t$ for describing the above basis vectors via the fixed bijection.

The $\alttest$ with parameter $t$ receives, as input, a pure state in $\hatt$ and performs a unitary operation in the Hilbert space $\mathcal{P}_t \otimes \hatt$, followed by a measurement. We will refer to the elements of $\mathcal{P}_t \otimes \hatt$ as consisting of two registers $R$ and $S$, 
where the content of each register is a mixed state with support over the corresponding Hilbert space.

Let us define $\ket{\mbox{perm}_t} = \frac{1}{\sqrt{t!}}\sum_{\pi\in S_t} \sgn(\pi)\ket{\pi}$.

\begin{center}
\fbox{
\begin{minipage}[l1pt]{6.00in}
{\bf  $\alttest(t)$} \\ 
Input: A pure state $\ket{\psi} \in \hatt$ in the $(t \cdot k)$-qubit register $S$\\
Output: The content of $S$ and  $\accept$ or  $\reject$.
\begin{enumerate}
	\item Create the state $(\frac{1}{\sqrt{t!}}\sum_{\pi\in S_t} \ket{\pi}) \otimes \ket{\psi}$. 
	\item Apply the unitary $U: \ket{\pi}\otimes \ket{\psi} \rightarrow \ket{\pi} \otimes U_\pi\ket{\psi}$.
	\item Perform the measurement $(M,I-M)$, where $ M=\ketbra{\mbox{perm}_t} \otimes I_\hatt$. Output the content of $S$. Output $\accept$ if the state has been projected onto $M$ and output $\reject$ otherwise.		
\end{enumerate}
\end{minipage}
}
\end{center}

It is easily verified that the $\alttest(t)$ runs in time polynomial in $t \cdot  k$. Since we will only call it with $t \in [q]$, its running time will be polynomial in $|x|$. The following  lemma states that the $\alttest(t)$ 
is a projection onto the subspace $\alt^\hatt$.

\begin{Lem}
\label{lem:alttest} 
\begin{enumerate}
\item For any pure state $\ket{\psi} \in \alt^\hatt$, the $\alttest(t)$ outputs the state $\ket{\psi}$ and $\accept$ with probability 1.	
\item 
For any $\ket{\phi} \in (\alt^\hatt)^\bot$, the $\alttest(t)$ outputs $\accept$ with probability 0.
\end{enumerate}
\end{Lem}

\begin{proof}
{ \bf Part 1: } 
Since  $\ket{\psi} \in \alt^\hatt$ we have 
$U_\pi \ket{\psi} = 
\sgn(\pi) \cdot \ket{\psi},$
and therefore the state after Step 2 is
\begin{eqnarray*}
\frac{1}{\sqrt{t!}}\sum_{\pi\in S_t} \ket{\pi} \otimes U_\pi \ket{\psi}  
=  \frac{1}{\sqrt{t!}}\sum_{\pi\in S_t} \sgn(\pi) \cdot \ket{\pi} \otimes \ket{\psi} 
 =  \ket{\mbox{perm}_t} \otimes \ket{\psi}.
\end{eqnarray*}

\noindent {\bf Part 2:} By Claim ~\ref{claim:Tdecompij}, we have $(\alt^{\hatt})^\bot  = \sum_{i \neq j} \sym^\hatt_{ij}$.
Hence it is enough to show that for all distinct $i,j \in [t]$ and for any vector $\ket{\phi} \in \sym^\hatt_{ij}$, the probability $p$ that the $\alttest(t)$ outputs $\accept$ is $0$. 
We have
\begin{eqnarray*}
p & = &  
( \frac{1}{\sqrt{t!}}\sum_{\sigma \in S_t}   \bra{\sigma}  \otimes \bra{\phi}  U^\dagger_{\sigma} ) ( \ketbra{\mbox{perm}_t} \otimes I_\hatt ) 
 ( \frac{1}{\sqrt{t!}}\sum_{\pi\in S_t} \ket{\pi} \otimes U_\pi \ket{\phi} ) \\
& = & (\frac{1}{t!})^2 \sum_{\sigma, \pi  \in S_t} \sgn(\sigma) \cdot \sgn(\pi) \cdot  \bra{\phi}  U^\dagger_{\sigma} U_\pi \ket{\phi} 
\end{eqnarray*}
We define $\pi' =  \pi \circ \pi_{ij}$. 
Then $\pi = \pi' \circ \pi_{ij}^{-1} = \pi' \circ \pi_{ij}$ and $\sgn(\pi) = -\sgn(\pi')$. Since $\ket{\phi} \in \sym^\hatt_{ij}$, we have $U_{\pi_{ij}} \ket{\phi} = \ket{\phi}$ and hence 
$U_{\pi' \circ \pi_{ij}} \ket{\phi} = U_{\pi'} \cdot (U_{\pi_{ij}} \ket{\phi}) = U_{\pi'} \ket{\phi}$. 
Therefore,
\begin{eqnarray*}
p
& = & (\frac{1}{t!})^2 \sum_{\sigma, \pi  \in S_t} \sgn(\sigma) \cdot \sgn(\pi) \cdot  \bra{\phi}  U^\dagger_{\sigma} U_\pi \ket{\phi} \\ 
& = & \frac{-1}{(t!)^2} \sum_{\sigma, \pi'  \in S_t} \sgn(\sigma) \cdot \sgn(\pi') \cdot  \bra{\phi}  U^\dagger_{\sigma} U_{\pi' \circ \pi_{ij}} \ket{\phi} \\
& = & \frac{-1}{(t!)^2} \sum_{\sigma, \pi'  \in S_t} \sgn(\sigma) \cdot \sgn(\pi') \cdot  \bra{\phi}  U^\dagger_{\sigma} U_{\pi'} \ket{\phi} \\
& = &  - p  
\end{eqnarray*}
Hence $p=0$.
\end{proof}

\subsection{Witness Test}
The $\wittest$ with parameter $t\in [q]$ receives as input a pure state in $\hatt$ and performs a unitary operation in the Hilbert space $\hatt \otimes \bit^{\otimes (t m)} $ followed by a measurement. We will  refer to the elements of $\hatt \otimes \bit^{\otimes (t m)}$ as consisting of $t$ pairs of registers $(T_i,Z_i)$ respectively on $k$ and $m$ qubits, for $i \in [t]$. All registers $Z_i$ will be initialized to $\ket{0^{m}}$. 

\begin{center}
\fbox{
\begin{minipage}[l]{6.00in}
{\bf $\wittest(t)$}\\
Input: A pure state $\ket{\psi} \in \hatt$  in the $k$-qubit registers $T_i$, for $i \in [t]$\\ 
Output: $\accept$ or $\reject$
\begin{enumerate}
\item For all $i \in [t]$, append a register $Z_i$ initialized to $\ket{0^m}$ and apply the circuit $V_x$ on registers $(T_i,Z_i)$. 
\item Output $\accept$ if for all $i \in [t]$, $V_x$ outputs $\accept$; otherwise output $\reject$.	
\end{enumerate}
\end{minipage}
}
\end{center}
We can describe the $\wittest(t)$ as the operator $(\Pi_{acc}V_x)^{\otimes t}$ acting on a state $\ket{\psi} \otimes \ket{0^{t m}}$. Hence,
$\prob[\wittest(t) \mbox{ outputs }\accept \mbox{ on } \ket{\psi}] =  || (\Pi_{acc}V_x)^{\otimes t} (\ket{\psi} \otimes \ket{0^{t m}})||^2 $.
Note that the description of the circuit $V_x^{\otimes t}$ can be generated in polynomial-time, since the circuit family $\{V_x, x \in \{0,1\}^*\}$ is uniformly generated  in polynomial-time. 

In what follows we will have to argue about the probability that the verification procedure $V_x$ outputs $\accept$ when its input is some mixed state. Even though we have only considered pure states as inputs in the definition of the class $\fqma$, we will see that it is not hard to extend our arguments to mixed states.

\begin{Lem}\label{lem:wittest}
\begin{enumerate}
\item If $x\in L_{yes}$, then for every $\ket{\psi} \in W_x^{\otimes t}$, 
the $\wittest(t)$ outputs $\accept$ with probability at least $2/3$.
\item If $x\in L_{yes}$, then for every $\ket{\phi} \in (W_x^{\otimes t})^\bot$, the $\wittest(t)$ outputs $\accept$ with probability at most $1/3$.
\item If $x\in L_{no}$, then for every $\ket{\psi} \in \hatt$, the $\wittest(t)$ outputs $\accept$ with probability at most $1/3$.
\end{enumerate}
\end{Lem}


\begin{proof}
 {\bf Part 1: } 
By completeness we know that for any pure state $\ket{\psi'} \in W_x$, we have \newline $\prob[V_x \mbox{ outputs } \reject \mbox{ on } \ket{\psi'}] \leq 2^{-r}$. Let $\rho_i$ denote the reduced density matrix of $\ket{\psi}$ on register $T_i$. Since $\ket{\psi} \in \wattx$, then for every $i \in [t]$, the density matrix $\rho_i$ is a distribution of pure states that all belong to $W_x$ and hence
$
\prob[V_x \mbox{ outputs } \reject \mbox{ on } \rho_i] \leq 2^{-r}.
$
It follows from the union bound that
\begin{eqnarray*} 
\prob[\wittest(t) \mbox{ outputs }\accept \mbox{ on } \ket{\psi}]
& \geq & 1 - \sum_{i=1}^t \prob[V_x \mbox{ outputs } \reject \mbox{ on } \rho_i] \\  
& \geq & 1-t\cdot 2^{-r} \geq 2/3,
\end{eqnarray*}
where the last inequality follows from the choice of $r$.

\vspace{0.1in}

\noindent{ \bf Part 2:} For $i \in [t]$, let 
$S_i = W_x \otimes \ldots \otimes W_x \otimes  W_x^{\bot} \otimes \h \ldots \otimes \h$,
where $W_x^{\bot} $ stands in the $i^{\mbox{\tiny th}}$ component of the tensor product.
By Fact~\ref{fact:stand} we have $(\wattx)^\bot = \bigoplus_{i\in [t]} S_i$.

Let us therefore consider a pure state $\ket{\phi} \in \bigoplus_{i\in [t]} S_i$ and let 
$\ket{\phi} = \sum_{i\in [t]} a_i \ket{\phi_i}$, where 
$\ket{\phi_i}$ is a pure state in $S_i$. 
Then $\sum_{i=1}^t |a_i|^2 = 1$ because the $\ket{\phi_i}$'s are also orthogonal. 
Furthermore, let $\rho_{i}$ be the reduced density matrix of $\ket{\phi_i}$ on register $T_i$. 
Then the support of $\rho_{i}$ is over $W_x^\bot$. Since for any pure state 
$\ket{\phi'} \in W_x^\bot$ we have 
$\prob[V_x \mbox{ outputs } \accept \mbox{ on } \ket{\phi'}] \leq 2^{-r}$,
we conclude that $ \prob[V_x \mbox{ outputs } \accept \mbox{ on } \rho_{i}] \leq 2^{-r}$.


The probability that the $\wittest(t)$ outputs $\accept$ on $\ket{\phi_i}$ is equal to the probability that all $t$ applications of $V_x$ output $\accept$, which is less than the probability that the $i^{\mbox{th}}$ application of $V_x$ outputs $\accept$, since the projections, performed in different registers, commute. Hence,
\begin{eqnarray*}
\prob[\wittest(t) \mbox{ outputs }\accept \mbox{ on } \ket{\phi_i}]
& \leq & \prob[V_x \mbox{ outputs } \accept \mbox{ on } \rho_{i}]
\quad \leq \quad 2^{-r}.
\end{eqnarray*}
Now for the input $\ket{\phi}$ we have
\begin{eqnarray*}
\prob[\wittest(t) \mbox{ outputs }\accept \mbox{ on } \ket{\phi}]
&  =  &
|| (\Pi_{acc}V_x)^{\otimes t} (\ket{\phi} \otimes \ket{0^{t  m}})||^2 \\  
& = & || \sum_{i\in [t]} a_i (\Pi_{acc}V_x)^{\otimes t} ( \ket{\phi_i} \otimes \ket{0^{t  m}})||^2 \\  
& \leq &  ( \sum_{i\in [t]} |a_i| \cdot||(\Pi_{acc}V_x)^{\otimes t} ( \ket{\phi_i} \otimes \ket{0^{t m}})|| )^2 \\
& \leq & 
 ( \sum_{i\in [t]} |a_i|^2 ) \cdot ( \sum_{i\in [t]} ||(\Pi_{acc}V_x)^{\otimes t} ( \ket{\phi_i} \otimes \ket{0^{t  m}})||^2 ) \\
 & \leq &  t\cdot 2^{-r} \;\; \leq \;\; 1/3 \enspace .
\end{eqnarray*}
In the above calculation the first two inequalities follow respectively from  the triangle inequality and the Cauchy-Schwarz inequality.

\vspace{0.1in} 

\noindent
{ \bf Part 3: }  
By the soundness of the original protocol we know that for any pure state $\ket{\psi'} \in \h$ 
$\prob[V_x \mbox{ outputs } \accept \mbox{ on } \ket{\psi'}] \leq 2^{-r}$. The same holds for any mixed state as well. Since the probability that the $\wittest(t)$ outputs $\accept$ is at most the probability that the procedure $V_x$ accepts the state on the first register $T_1$ we conclude that 
\begin{eqnarray*}
\prob[\wittest(t) \mbox{ outputs }\accept \mbox{ on } \ket{\psi}]
 \leq  \prob[V_x \mbox{ outputs } \accept \mbox{ on } \rho_1] 
 \leq  2^{-r} \quad \leq \quad 1/3.  \;\;\;\;\;\;\;\;\quad\quad\quad\quad\quad\quad\quad\quad\quad\quad\quad \Box 
\end{eqnarray*}
\end{proof}


\subsection{Putting it all together}

Finally we describe the algorithm $\mcA$ in the figure below and proceed to analyze its properties.\\

\noindent {\bf Running time:} 
We have seen that the description of the circuit that performs the $\alttest(t)$ can be generated in time polynomial in $t \cdot k$ which is polynomial in $|x|$. The description of the circuit that performs the $\wittest(t)$ can also be generated in polynomial-time, since the circuit family $\{V_x: x \in \{0,1\}^*\}$ can be generated uniformly in polynomial-time. Hence the description of the circuit $A_x^t$ can be generated in polynomial-time and the overall algorithm $\mcA$ runs in polynomial-time.

\begin{center}
\fbox{
\begin{minipage}[l]{6.00in}
\begin{center}Algorithm $\mcA$\end{center}
 Input: $x \in L_{yes} \cup L_{no}$\\ 
Output: $\accept$ or $\reject$


\begin{enumerate}
	\item For $t = 1, \ldots, q(|x|)$ do:
		\begin{enumerate}
			\item Call the oracle $\mcO$ with input $A_x^t$, where $A_x^t$ is the description of the circuit of the following procedure on $t\cdot k$ witness qubits and $t \cdot m$ auxiliary qubits

\begin{center}
\begin{minipage}[l]{4.00in}
Input: A pure state $\ket{\psi} \in \hatt \;\; ;\;\;$
Output: $\accept$ or $\reject$
\begin{enumerate}
	\item Run the $\alttest (t)$ with input $\ket{\psi}$. 
	\item Run the $\wittest (t)$ with input being the output state of the $\alttest(t)$. 
	\item Output $\accept$ iff both Tests output $\accept$.			
\end{enumerate}
\end{minipage}
\end{center}
			\item If $\mcO$ outputs $\accept$ then output $\accept$ and halt.
		\end{enumerate}
\item Output $\reject$.
\end{enumerate}
\end{minipage}
}
\end{center}

\COMMENT{
\begin{figure}[ht]
\label{fig:algorithm}
\begin{center}
\fbox{
\begin{minipage}[l]{5.50in}
 
Input: $x \in L_{yes} \cup L_{no}$. Output $ \in \{\accept, \reject\}$.


Let\footnote{Recall that $V_x$ have $k(|x|)$ input qubits and $m(|x|)$ axillary qubits.} $\h_1 \defeq \cc^{\otimes k(|x|)}$ and $\h_2 \defeq \cc^{\otimes m(|x|)}$. 
\begin{enumerate}
	\item For $d = 1 \ldots q(|x|)$ do:
		\begin{enumerate}
			\item Let ${\tilde V}^d_{\alt}$ be the circuit of the $\alttest$ for the Hilbert space $\h_1^{\otimes d}$. Let $V^d_{\alt}$ be the circuit which acts like ${\tilde V}^d_{\alt}$ on $\h_1^{\otimes d}$ and identity on $\h_2^{\otimes d}$. Generate description of the circuit $V^d_{\alt}$. 
			\item Let $V_x^d$ represent $d$ parallel copies of the circuit $V_x$ such that $V_x^d$ outputs $\accept$ iff all $d$-copies of $V_x$ output $\accept$. Generate description of the circuit $V^d_x$. 
			\item Let $W^d_x$ be the composition of $V^d_{\alt}$ followed by $V_x^d$ such that $W^d_x$ output $\accept$ iff $V^d_{\alt}$ output $\accept$ followed by $V_x^d$ output $\accept$. Generate description of circuit $W_x^d$. 
			\item Call oracle $\mcO$ with input $(W_x^d, d \cdot k(|x|))$. 
			\item If $\mcO$ output $\accept$ then output $\accept$ and halt. If $\mcO$ output $\reject$ then set $d = d +1$.  
		\end{enumerate}
\item Output $\reject$ and halt.
\end{enumerate}
\end{minipage}
}
\end{center}
\caption{Algorithm $\mcA$}
\end{figure}
}

\vspace{0.05in}

\noindent{\bf Correctness in  case $x \in L_{yes}$:} Let us consider the oracle call with input $A_x^d$ where $d$ is the dimension of $W_x$. We prove that $ A_x^d \in \uqmacp_{yes}$, hence the oracle $\mcO$ outputs $\accept$  and therefore $\mcA$ outputs $\accept$ as well. Our claim is immediate from the following lemma.

\begin{Lem}
\begin{enumerate}
\item $\Pr[A_x^d \mbox{ outputs } \accept \mbox{ on } \ket{W_{alt}}] \geq 2/3$. 
\item Let $\ket{\phi} \in \hd$ be orthogonal to $\ket{W_{alt}}$. Then $\Pr[A_x^d \mbox{ outputs } \accept \mbox{ on } \ket{\phi}] \leq 1/3$.
\end{enumerate}
\end{Lem}

\begin{proof}
{\bf Part (1): } 
Since $\ket{W_{alt}}\in \alt^{\hd}$, Lemma \ref{lem:alttest} tells us that the $\alttest(d)$ outputs the state $\ket{W_{alt}}$ and $\accept$ with probability 1. Then, since for every $i \in [d] $ the support of the reduced density matrix of $\ket{W_{alt}}$ on register $T_i$ is on $W_x$, Lemma \ref{lem:wittest} tells us that the $\wittest$ outputs $\accept$ with probability at least $2/3$. 

\vspace{0.05in}

\noindent {\bf Part (2):} By Claim~\ref{claim:projection} and the fact that the state $\ket{\phi}$ is orthogonal to $\ket{W_{alt}}$, we can conclude that if the $\alttest(d)$ outputs $\accept$ then the output state is a pure state $\ket{\phi'} \in (\watd )^\bot$. 
Now, by Lemma \ref{lem:wittest}, the probability that the $\wittest(d)$ outputs $\accept$ on input $\ket{\phi'}$ is at most $1/3$. 
\end{proof}

\noindent{\bf Correctness in  case $x \in L_{no}$:} By Lemma \ref{lem:wittest} it follows easily that for all $t \in [q]: \;  A_x^t \in \uqmacp_{no}$. In this case, $\mcO$ outputs $\reject$ in every iteration, and hence $\mcA$ outputs $\reject$. \\

\noindent This concludes the proof of Theorem \ref{main}.

\end{proof}
\suppress{
\subsection*{Open questions}
There are two main open questions concerning this work.
\begin{enumerate}
\item As defined, for $L \in \fqma$, there is a verification procedure such that for $x \in L_{yes}$ there are
no eigenvalues of the operator $\Pi_x$ in the open interval
$(2/3,1/3)$. Can this be assumed without loss of generality for any $L
\in \qma$? That is,
is there some way to ensure that for any $L \in \qma$, there is a
verification procedure such that for all $x \in L$, the corresponding
operator $\Pi_x$ has no eigenvalue in the open interval $(2/3,1/3)$ ?
\item For $L \in \qma$, assume that it has a verification procedure
such that for all $x \in L_{yes}$, there are no eigenvalues of $\Pi_x$ in the
open interval $(2/3,1/3)$, however there could be exponentially many
eigenvalues more than $2/3$. Is there some way to efficiently
(Turing) reduce $L$ to $\uqmacp$, possibly via probabilistic or
quantum procedures ?  
\end{enumerate}
}
\subsection*{Acknowledgments} 
Most of this work was conducted when R.J., I.K., M.S. and S.Z. were at the Centre for Quantum Technologies (CQT) in Singapore, and partially funded by the Singapore Ministry of Education and the National Research Foundation. Research partially supported by European Commission IST projects Qubit Application (QAP) 015848 and Quantum Computer Science (QCS) 25596, by the French ANR programs under contract ANR-08-EMER-012 (QRAC project) and ANR-09-JCJC-0067-01 (CRYQ project), by the French MAEE STIC-Asie program FQIC and by NSF grant DMS-0606795 . R.J., I.K., M.S. and S.Z. would like to thank Hartmut Klauck for several insightful discussions during the process of this work. G.K. would like to thank Scott Aaronson and Dorit Aharonov for discussions.



\appendix

\section{Yet another definition of {\sf FewQMA} }
\label{sec:third}

We have seen two definitions for $\fqma$ and have
proven their equivalence. In high level, one says that in the yes
instances there is a polynomial number of eigenvalues of the
projection operator $\Pi_x$ that are larger than $2/3$, while no
eigenvalue is in the interval $(1/3,2/3)$. The other definition says
that in a yes instance there exists a subspace of polynomial dimension
such that every state in the subspace is accepted with probability at
least $2/3$ and every state orthogonal to this subspace is accepted
with probability at most $1/3$. 

A natural question is whether the use of a subspace is necessary in
the second definition or we could have just talked about a set of
orthonormal vectors of polynomial size, where each vector is accepted
with probability $2/3$ and every vector orthogonal to these ones is
accepted with probability at most $1/3$. While we are unable to
show the equivalence of the complexity class defined this way and $\fqma$,
a weak equivalence can indeed be shown. For this we include the parameters of the
verification procedure and the bound on the number of witnesses in the definition of the class,
and we also require strong amplification.
More precisely, consider the following definition.  
  
\begin{Def}\label{def:fqma3}
Let $c,w,s: \mathbb{N} \rightarrow [0,1]$ be polynomial time
computable functions such that $c(n) > \max\{w(n), s(n)\}$ for all $n
\in \mathbb{N}$. Let $q,k,m$ be polynomials such that $q(n) \leq 2^{k(n)}$ for all $n
\in \mathbb{N}$. A promise problem $L=(L_{yes}, L_{no})$ is in the
complexity class  
{\em Vector Few Quantum Merlin-Arthur}
$\vfqma(c,w,s,q,k,m)$ if there exists a verification procedure
$\{V_x : x \in \{0,1\}^*\}$ with $k$ witness qubits and  $m$ ancilla
qubits, such that
\begin{enumerate}
\item 
for all $x \in L_{yes}$ there exists an orthonormal basis $\{
\ket{\psi_1}, \ldots, \ket{\psi_{2^k}}\}$ of the witness space and $d \in [q(|x|)]$ such
that 
\begin{enumerate} 
\item for all pure states $\ket{\psi_i}$ with  $i\in [d],\;\;$
$|| \Pi_{acc} V_x (\ket{\psi_i} \otimes \ket{0^{m}}) ||^2 \geq c(|x|),$
\item for all pure states $\ket{\psi_i}$ with  $d+1 \leq i \leq 2^k,\;\;$ 
$|| \Pi_{acc} V_x (\ket{\psi_i} \otimes \ket{0^{m}}) ||^2 \leq w(|x|),$
\end{enumerate}
\item
for all $x \in L_{no}$ and for all pure states
$\ket{\psi} \in \bit^{\otimes k}$,
$|| \Pi_{acc} V_x (\ket{\psi} \otimes \ket{0^{m}}) ||^2 \leq s(|x|).$
\end{enumerate}
Finally we define 
$$\vfqma = 
 \bigcup_{ q,k,m: ~ q \leq 2^k}
\hspace{-0.2in}  \vfqma(1- \frac{1}{3q}, \frac{1}{3 \cdot 2^k}, \frac{1}{3}, q, k, m)
$$
\end{Def}


We show $\vfqma = \fqma$ by using Horn's
Theorem that states that for a Hermitian matrix, the vector of the eigenvalues majorizes the diagonal. 
\begin{Thm}{\em \cite{H54}}
Let $R$ be a natural number. Let $\Lambda=\{\lambda_1 \geq \lambda_2
\geq \ldots \geq \lambda_R | \lambda_i \in \mbR \}$ and $A = \{ \mu_1
\geq \mu_2\geq \ldots \geq \mu_R | \mu_i \in \mbR\}$.  Then there
exists an $R \times R$ Hermitian matrix with set of eigenvalues
$\Lambda$ and set of diagonal elements $A$ if and only if
$\sum_{i=1}^t(\lambda_i-\mu_i) \geq 0$ for all $t \in [R]$ and with
equality for $t = R$.
\end{Thm}
\begin{Thm}\label{thm:equivalence2}  
$
\fqma = \vfqma \enspace .
$ 
\end{Thm}

\begin{proof}
We first show $\fqma \subseteq \vfqma$. Let $L \in \fqma$ have a verification procedure $\{V'_x : x \in
\{0,1\}^*\}$ with polynomials $k,m',q$. Let $r$ be a polynomial such
that $2^{-r} \leq \frac{1}{3 \cdot 2^k}$ and $2^{-r} \leq \frac{1}{3
q}$. We know from Theorem \ref{thm:amplify} that $L \in \fqma(1
- 2^{-r}, 2^{-r}, 2^{-r})$ with verification procedure $\{ V_x : x \in
\{0,1\}^* \}$ and polynomials $k,m,q$, where $m = poly(m',  r)$. 
Then the eigenbasis of $\Pi_x$ satisfies the requirements of
the definition of $\vfqma(1-\frac{1}{3q},\frac{1}{3 \cdot
2^k},\frac{1}{3},q,k,m)$.  This shows that $L \in \vfqma$.  

We now show $\vfqma \subseteq \fqma$. Let $L \in \vfqma(1-\frac{1}{3q},\frac{1}{3 \cdot 2^k},\frac{1}{3},q,k,m)$, for
some polynomials $q,k,m$ such that $q \leq 2^k$. Let $N = 2^k$ and
$\h = \bit^{\otimes k}$. 
If $x \in L_{yes}$, then  there exist an
orthonormal basis $\{ \ket{\psi_1}, \ldots, \ket{\psi_N} \}$
for $\h$ and $d \in [q]$, such that for $i \in [d]$, $ \mu_i \geq 1-\frac{1}{3q}$ and for $d+1 \leq i \leq N$, $\mu_i
\leq  \frac{1}{3 N}$, where by definition
\begin{align*} 
\mu_i & \defeq || \Pi_{acc} V_x (\ket{\psi_i} \otimes \ket{0^{m}}) ||^2  
  = \bra{\psi_i}
\otimes \bra{0^m} \Pi_x \ket{\psi_i} \otimes
\ket{0^m}
\end{align*}

Consider now the Hermitian matrix $M$ that describes the projection
operator  $\Pi_x$ in a  basis that is an extension of $\{ \ket{\psi_1}
\otimes \ket{0^m}, \ldots, \ket{\psi_N} \otimes \ket{0^m} \}$. Note
that $\mu_i, i \in [N]$ are the first $N$ diagonal elements of $M$.  Observe that an eigenvector of $\Pi_x$ with non-zero eigenvalue
is also an eigenvector of $\Pi_{init}$ with non-zero eigenvalue. Since there are $N$
non-zero eigenvalues of $\Pi_{init}$, there are at most $N$
non-zero eigenvalues of $\Pi_x$. This also implies that $\mu_i =0$ for
$N < i \leq 2^{k+m}$. Let the first $N$ eigenvalues of $\Pi_x$ in decreasing
order be $\lambda_i, i \in [N]$.  Then, using Horn's theorem,
\[ \sum_{i=1}^d \lambda_i \geq \sum_{i=1}^d \mu_i \geq d \cdot (1-\frac{1}{3q})
\] 
which implies (since $\lambda_i \leq 1$, for all
$i \in [N]$) that 
$$\lambda_d \geq -(d-1)+ d \cdot (1-\frac{1}{3q}) \geq 1-\frac{d}{3q}
\geq \frac{2}{3} \enspace .$$
Also, we have that 
\begin{eqnarray*}
\lambda_{d+1} & \leq & \sum_{i=d+1}^{2^{k+m}} \lambda_i \quad \leq \quad
\sum_{i=d+1}^{2^{k+m}} \mu_i \quad = \quad \sum_{i=d+1}^{2^{k}} \mu_i 
\quad \leq \quad (N-d)\frac{1}{3 N} \quad \leq \quad \frac{1}{3} \enspace .
\end{eqnarray*}
If $x \in L_{no}$ then by the soundness condition
$ \lambda_{1} \leq \frac{1}{3}.$ This shows that $L \in \fqma$.
\end{proof}

\COMMENT{
This Theorem shows that in fact the definition of $\fqma$ with vectors instead of a subspace is not very robust. Note that Horn's theorem is tight in the sense that we can have Hermitian matrices where the relationship between the eigenvalues and the diagonal elements thereof are the worst possible for our case. In other words, the requirements for the parameters $(c,w,d)$ in the class $\vfqma$ in order for it to be equivalent to the class $\fqma$ are tight and of course rather severe. Hence, it seems that the definition of $\fqma$ with the subspace captures better the essence of multiple quantum witnesses.
}

\end{document}